%% file: main.tex
\documentclass[letterpaper, 10 pt, conference]{ieeeconf}  
\IEEEoverridecommandlockouts                              
\usepackage{epsfig} 
\usepackage{graphicx} 
\usepackage{amsmath} 
\usepackage{amssymb} 
 
\usepackage{amsthm} 
\usepackage{enumerate}
\usepackage{bbm}
\usepackage{mathtools, cuted}
\usepackage{xcolor}
\usepackage{mathrsfs}
\usepackage{cite}
\allowdisplaybreaks

\newtheorem{assumption}{\bf Assumption}

\newtheorem{theorem}{\bf Theorem}
\newtheorem{lemma}{\bf Lemma}
\newtheorem{proposition}{\bf Proposition}

\newtheorem{corollary}{\bf Corollary}
\theoremstyle{definition}

\newtheorem{remark}{\bf Remark}

\newcommand{\norm}[1]{\ensuremath{\left\| #1\right\|}}

\newcommand{\paren}[1]{\ensuremath{\left( #1\right)}}
\newcommand{\clint}[1]{\ensuremath{\left[ #1\right]}}
\newcommand{\set}[1]{\ensuremath{\left\{ #1\right\}}}
\newcommand{\matr}[1]{\ensuremath{\clint{\begin{array} #1 \end{array}}}}
\newcommand{\R}{\ensuremath{\mathbb{R}}}
\newcommand{\E}{\ensuremath{\mathbb{E}}}
\newcommand{\F}{\ensuremath{\mathscr{F}}}
\newcommand{\LL}{\ensuremath{\mathcal{L}}}
\DeclareFontFamily{OT1}{pzc}{}
\DeclareFontShape{OT1}{pzc}{m}{it}{<-> s * [1.200] pzcmi7t}{}
\DeclareMathAlphabet{\mathpzc}{OT1}{pzc}{m}{it}
\newcommand{\LAG}{\ensuremath{{\mathpzc{L}_{\hspace{.1pt}}}}}

\newcommand{\Neg}[1]{\hspace{- #1 bp}}

\DeclareMathOperator{\Tr}{\mathrm{tr}}

\title{\LARGE \bf
Risk-Constrained Linear-Quadratic Regulators\vspace{-5bp}}

\author{Anastasios~Tsiamis, Dionysios~S.~Kalogerias, Luiz~F.~O.~Chamon,\\ Alejandro~Ribeiro and George~J.~Pappas
\thanks{The authors  are   with   the   Department   of   Electrical   and   Systems  Engineering,  University  of  Pennsylvania,  Philadelphia,  PA  19104 (email: \{atsiamis,dionysis,luizf,aribeiro,pappasg\}@seas.upenn.edu). }
\thanks{This work is supported by the AFOSR under grant FA9550-19-1-0265 (Assured Autonomy in Contested Environments), and the NSF under grant CPS 1837253.
}
}

\begin{document}

\maketitle
\thispagestyle{empty}
\pagestyle{empty}

\begin{abstract}
We propose a new risk-constrained reformulation of the standard Linear Quadratic Regulator (LQR) problem. Our framework is motivated by the fact that the classical (risk-neutral) LQR controller, although optimal in expectation, might be ineffective under relatively infrequent, yet statistically significant (risky) events.
To effectively trade between average and extreme event performance, we introduce a new risk constraint, which explicitly restricts the total expected predictive variance of the state penalty by a user-prescribed level.
We show that, under rather minimal conditions on the process noise~(i.e., finite fourth-order moments), the optimal risk-aware controller can be evaluated explicitly and in closed form. In fact, it is affine relative to the state, and is always internally stable regardless of parameter tuning. Our new risk-aware controller:
i) pushes the state away from directions where the noise exhibits heavy tails, by exploiting the third-order moment~(skewness) of the noise; ii) inflates the state penalty in riskier directions, where both the noise covariance and the state penalty are simultaneously large. 
The properties of the proposed risk-aware LQR framework are also illustrated via indicative numerical examples.
\end{abstract}
\input{Introduction}
\input{Formulation}

\input{Analysis}

\input{Simulations}

\input{Extensions}

\input{appendix}
\bibliographystyle{IEEEbib-abbrev.bst}
\bibliography{IEEEabrv,risk_literature,library_fixed}
\end{document}

%% file: Introduction.tex
\section{Introduction}\label{Section_Introduction}
Achieving good performance in expectation is often insufficient in 
the design of stochastic control systems, 
especially when dealing with modern, critical applications. Examples
appear naturally in many areas, including wireless industrial control \cite{Ahlen2019}, energy \cite{Bruno2016,Moazeni2015}, finance \cite{Markowitz1952,Follmer2002,Shang2018}, robotics \cite{Kim2019,Pereira2013}, networking \cite{Ma2018}, and safety \cite{samuelson2018safety,chapman2019cvar}, to name a few.  Indeed, occurrence of less probable, non-typical or unexpected events might lead the underlying dynamical system to experience shocks with possibly catastrophic consequences, e.g., a drone diverging too much from a given trajectory in a hostile environment, or an autonomous vehicle crashing onto a wall or hitting a pedestrian. In such situations, design of effective \emph{risk-aware} control policies is highly desirable, systematically compensating for those extreme events, at the cost of slightly sacrificing average performance under nominal conditions.

To highlight the usefulness of a risk-aware control policy, let us consider the following simple, motivating example. Let $x_{k+1}=x_{k}+u_k+w_{k+1}$ model an aerial robot, moving along a line. Assume that the process noise $w_k$ is \textit{i.i.d.} Bernoulli, taking the values $\beta>2$ with probability $1/\beta$ and $0$ with probability $1-1/\beta$. This noise represents shocks, e.g., wind gusts, that can occur with some small probability. We would then like to minimize the LQR cost $\E \sum_{t=0}^{N} \{x^2_t\}$, i.e., the total displacement of the robot over a horizon of $N$ time steps. In this case, the LQR optimal controller is $u^{\mathrm{LQR}}_k=-x_k-1$, where $-1\equiv-\mathbb{E}w_k$ cancels the mean of the process noise. We see that the LQR solution is \textit{risk-neutral}, as it does not account for the fact that the shock $\beta$ could be arbitrarily large.  On the other hand, the risk-aware LQR formulation proposed in this work results in a \textit{family} of optimal controllers of the form
\[
u^{*}_{t}(\mu)=-x_t-1-\frac{\mu}{1+2\mu}(\beta-2),\quad\mu\ge0,
\]
where $\mu$ controls the trade-off between average performance and risk. As $\mu$ increases, we move from the risk-neutral to the \textit{maximally risk-aware} controller $u^{*}_{t}(\infty)=-x_t-\beta/2$, which treats the noise as adversarial---see Fig.~\ref{fig:toy_example}.
\begin{figure}[t]
\vspace{6.5bp}
	\centering
	\includegraphics[width=\columnwidth]{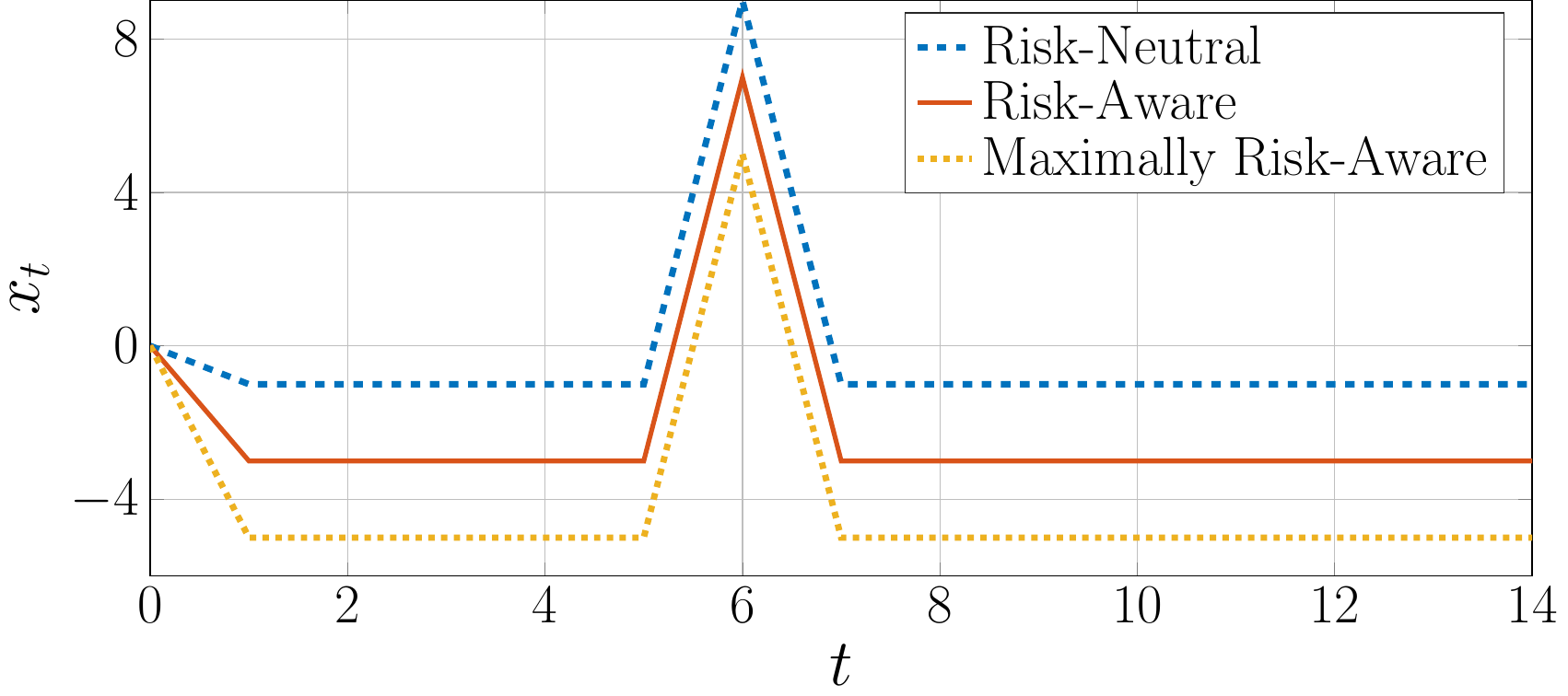}\vspace{-4bp}
	\caption{Comparison between risk-neutral and risk-aware control performance, when the system experiences rare but large shocks---here the shock occurs at time $6$. By sacrificing average behavior, the risk-aware controllers push the state away from the direction of the shock. 
}
	\label{fig:toy_example}
 	\vspace{-13bp}
\end{figure}

In both classical and recent literature in linear-quadratic problems, risk awareness in estimation and control is typically achieved by replacing the respective random cost with its exponentiation~\cite{jacobson1973optimal,Whittle1981,bacsar2000risk,pham2012linear,roulet2019convergence,Speyer1992,Dey1999,Moore1997,Dey1997,Bauerle2014more}. Yet, the resulting stochastic control problem might not be well-defined for general classes of noise distributions, as it requires the moment generating function of the cost to be finite. Thus, heavy-tailed or skewed distributions, which are precisely those exhibiting high risk, are naturally excluded. Also, even if the expectation of the exponential cost is finite, it does not lead to a general, closed-form and interpretable solution. 
A notable exception is that of a
Gaussian process noise, also known as the Linear Exponential Quadratic Gaussian (LEQG) problem, which does enjoy a simple closed-form solution~\cite{Whittle1981}. 
Apparently though, the Gaussian assumption is unable to capture distributions with asymmetric (skewed) structure, as in the above example.

In this paper, we propose a new risk-aware reformulation of the LQR problem, in which the standard LQR objective is minimized \textit{subject to} an explicit and tunable risk constraint. Our contributions are as follows.

\noindent\textbf{--New Risk Measure.} We introduce  the cumulative expected one-step predictive variance of the associated state penalty as a new risk measure for LQR control.
In this way, our risk-constraint formulation ensures not only a small LQR cost, but also \textit{guaranteed} statistical variability of the state penalty.

\noindent
\textbf{--Optimal Risk-Aware Controls \& Stability.}
We show that our new risk-constrained formulation results in a quadratically constrained LQR problem (Proposition~\ref{ANA_PROP_Reformulation}), which admits an explicit closed-form solution with a natural interpretation. The optimal risk-aware feedback controller is affine with respect to the system state (Theorems~\ref{ANA_THM_Optimal_Input} \& \ref{OPTIMAL}). 
The affine component of the control law 
pushes the state away from directions where the noise exhibits (skewed) heavy tails. Meanwhile, the state feedback gain satisfies a new risk-aware Riccati recursion, in which the state penalty is inflated in riskier directions, where both the noise covariance and state penalty are simultaneously larger. Further, we show that our optimal risk-aware controller is always stable, under standard LQR conditions (Corollary \ref{STABILITY}).

\noindent\textbf{--Arbitrary Noise Model.}
Contrary to the LEQG approach, our results are valid for arbitrary noise distributions, provided the associated fourth-order moments are finite; thus, heavy-tailed or skewed noises are supported within our framework.

\noindent\textbf{--Relation to LQR with Tracking.}
We show that, \textit{by appropriate re-parameterization}, our risk-aware LQR problem is equivalent to a generalized risk-neutral LQR problem with a tracking objective. Essentially, this implies that risk-neutral LQR formulations can provide inherent risk-averse behavior, as long as the involved parameters are selected in a principled way, as presented herein.

\textbf{\textit{Related Work:}}
Risk-aware optimization has been studied in a wide variety of decision making contexts \cite{A.2018,Cardoso2019,W.Huang2017,Jiang2017,Kalogerias2018b,Tamar2017,Vitt2018,Zhou2018,Ruszczynski2010,SOPASAKIS2019281,chapman2019cvar}. The basic idea is to replace expectations by more general functionals, called risk measures\cite{ShapiroLectures_2ND}, purposed to effectively quantify the statistical volatility of the involved random cost function, in addition to mean performance. Typical examples are mean-variance functionals \cite{Markowitz1952,ShapiroLectures_2ND},
mean-semideviations \cite{Kalogerias2018b}, and Conditional Value-at-Risk
(CVaR) \cite{Rockafellar1997}.

In the case of control systems, apart form the aforementioned exponential approach, CVaR optimization techniques have also been considered for risk-aware constraint satisfaction~\cite{chapman2019cvar}. Although CVaR captures variability and tail events well, CVaR optimization problems rarely enjoy closed-form expressions. Approximations are usually required to make computations tractable, e.g., process noise and controls are assumed to be finite-valued~\cite{chapman2019cvar}, thus excluding the LQR setting. Predictive variance constraints have also been used as a measure of risk in portfolio optimization~\cite{abeille2016lqg}. Different from our paper, the noise is assumed to be Gaussian and the variance is with respect to linear stage costs.

Another related concept is that of robust control, where the system model is unknown. The objective is to optimally control the true system under worst case model uncertainty. When there is no model uncertainty, in the case of stochastic process noise, robust controllers usually reduce to their risk-neutral LQR counterparts~\cite{tzortzis2016robust,dean2017sample}. On the contrary, in risk-aware control, extreme noise events are part of the system model. Even if the system is exactly modeled, we would still need to consider risk-aware control if the process noise is heavy-tailed or highly variable. From this point of view, robustness and risk are orthogonal concepts.
Interestingly, in the case of adversarial noise, there is a connection between robust and maximally risk-aware controllers~\cite{glover1988state}.

\textit{\textbf{Notation:}}
The transpose operation is denoted by $(\cdot)'$. If $x_{k},\dots,x_{t}$ is a sequence of vectors, then $x_{k:t}$ denotes the batch vector of all $x_i$ for $k\le i\le t$. The $\sigma$-algebra generated by a random vector $x$ is denoted by $\sigma(x)$.

%% file: Formulation.tex
\section{Risk-Constrained LQR Formulation} \label{Section_Formulation}

Consider a discrete-time linear system in state-space form evolving according to the stochastic difference equation
\begin{equation}\label{FOR_EQN_system}
x_{t+1}=Ax_{t}+Bu_{t}+w_{t+1},
\end{equation}
where $x_t\in \R^n$ is the state, $u_t\in \R^p$ is an exogenous control signal, $A$ is the state transition matrix, and $B$ is the input matrix. We assume that the initial value $x_0$ is deterministic and fixed. Signal $w_{t}\in \R^n$ is a random process noise (not necessarily Gaussian) and is assumed to be \textit{i.i.d}. 
For $t\ge0$, let $\F_t=\sigma\paren{x_{0:t},u_{0:t}}$ be the $\sigma$-algebra generated by all observables up to time $t$, and let $\F_{-1}$ be the trivial $\sigma$-algebra. Based on this notation, $x_t,u_t,w_t$ are $\F_t$-measurable, while $w_{t+1}$ is independent of $\F_t$.
We also make an additional assumption on the process noise, as follows.
\begin{assumption}[\textbf{Noise Regularity}]\label{FOR_ASS_noise}
The process $w_t$ has finite fourth-order moment, i.e., for every $t\in \mathbb{N}$, $\E\norm{w_{t}}^4_2<\infty$.
\end{assumption}
\noindent The above condition is mild and satisfied by very general noise distributions, including many heavy-tailed ones.
Denote the mean of the noise by $\bar{w}\triangleq\E w_k$ and its variance by $W\triangleq \E(w_k-\bar{w})(w_k-\bar{w})'$.

 In the classical risk-neutral formulation of the LQR problem, one is interested in the multistage stochastic program
 \begin{equation}\label{FOR_EQN_Risk_Neutral_LQR}
\begin{aligned}
\min_{u}&\quad \E\set{ x'_NQ x_N+\sum_{t=0}^{N-1} x'_tQx_t+u'_tRu_t}\\
\mathrm{s.t.}&\quad x_{t+1}=Ax_t+Bu_t+w_{t+1}\\
&\quad u_t\in \LL_{2}(\F_t),\,t=0,\dots N-1
\end{aligned},
\end{equation}
where $u=u_{0:N-1}$ are the inputs from time $0$ up to time $N-1$, for some horizon $N\in \mathbb{N}$. For each $t$, the \textit{causality constraint} on $u_t$ restricts the inputs to the space of square-integrable $\F_t$-measurable vector-valued random elements of appropriate dimension,
denoted as $\LL_2(\F_t)$. It also guarantees that the optimization problem is well-defined and with finite cost. In order for the optimal LQR controller to be well-behaved and stable as the horizon $N$ grows, we also make the following standard assumption.
\begin{assumption}[\textbf{LQR}]\label{FOR_ASS_LQR}
The pair $(A,B)$ is stabilizable, the pair $(A,Q^{1/2})$ is detectable, matrix $Q\succeq 0$ is positive semi-definite and matrix $R\succ 0$ is positive definite.
\end{assumption}
As mentioned above, the classical LQR problem is risk-neutral, since it optimizes performance only on average ~\cite{bertsekas2017dynamic}. 
Still, even if the average performance is good, the state can grow arbitrarily large under less probable, yet extreme events.
In other words, the state can exhibit large variability.
To deal with this issue, we propose a risk-constrained formulation of the LQR problem, posed as
\begin{equation}\label{FOR_EQN_LQR_constrained}
\begin{aligned}
\min_{u} &\quad \E\set{ x'_NQ x_N+\sum_{t=0}^{N-1} x'_tQx_t+u'_tRu_t}\\
\mathrm{s.t.} &\quad \E \set{\sum_{t=1}^{N} \left[ x'_tQ x_t-\E\paren{x'_tQx_t|\F_{t-1}} \right]^2}\le \epsilon\\
&\quad x_{t+1}=Ax_t+Bu_t+w_{t+1}\\ 
&\quad u_t\in \LL_{2}(\F_t),\,t=0,\dots, N-1
\end{aligned}\,.
\end{equation}
\noindent  Here, the risk measure adopted is the cumulative expected predictive variance of the state cost. The predictive variance incorporates information about the tail \textit{and} skewness of the penalty $x_t'Qx_t$. This forces the controller to take higher-order noise statistics into account, mitigating the effect of rare though large noise values.
Hence, our risk-aware LQR formulation not only forces the state $x_t$ to be close to zero, but also explicitly restricts its variability.

Problem \eqref{FOR_EQN_LQR_constrained} offers a simple and interpretable way to control the trade-off between average performance and risk. By simply decreasing $\epsilon$, we increase risk-awareness. Inspired by standard risk-aware formulations, in the above optimization problem our risk definition is tied to the specific state penalty $x'_tQx_t$ of the LQR. However, all of our results are still valid if we employ the predictive variance of a different quadratic form, e.g., the norm of the state, $\norm{x_t}^2$, in the constraint---see Section~\ref{Section_Extensions}. 
Lastly, note that the initial state is fixed (for simplicity), so there is no associated risk term for $t=0$.

In the next section, we show that the risk constraint of \eqref{FOR_EQN_LQR_constrained} can be rewritten in quadratic form. This will allow us to solve problem \eqref{FOR_EQN_LQR_constrained} using duality theory and obtain a closed-form solution, exploiting higher-order noise moments.

%% file: Analysis.tex
\section{Optimal Risk-Aware LQR Controllers
}\label{Section_Analysis}

The analysis of the risk-aware dynamic program~\eqref{FOR_EQN_LQR_constrained} consists of the following steps. First, we ensure the well-definiteness of~\eqref{FOR_EQN_LQR_constrained}, also showing that~\eqref{FOR_EQN_LQR_constrained} can be equivalently reexpressed as a \textit{sequential variational Quadratically Contrained Quadratic Program (QCQP)}, or, more precisely, as a \textit{Quadratically Constrained LQR (QC-LQR)} problem (Proposition~\ref{ANA_PROP_Reformulation}). Then, we exploit Lagrangian duality (Theorem~\ref{KKT}) to solve~\eqref{FOR_EQN_LQR_constrained} \textit{exactly and in closed form}. More specifically, we first derive an explicit expression for the optimal risk-aware controller~(Theorem~\ref{ANA_THM_Optimal_Input}), given an arbitrary but fixed Lagrange multiplier. Then, we show how an optimal Lagrange multiplier may be efficiently discovered via trivial bisection~(Theorem~\ref{OPTIMAL} and Proposition \ref{ANA_PROP_Risk_Evaluation}).
\begin{proposition}[\textbf{QCQP Reformulation}]\label{ANA_PROP_Reformulation}
	Let Assumption~\ref{FOR_ASS_noise} be in effect, and define the higher-order weighted statistics
	\begin{align}
	M_3&\triangleq \E\set{(w_{i}-\bar{w})(w_{i}-\bar{w})'Q(w_{i}-\bar{w})}\in\mathbb{R}^n
	\,\,\, \textrm{and}\nonumber\\
	m_4&\triangleq \E\Neg{-.5}\big\{\Neg{1.5}\left[(w_{i}-\bar{w})'Q(w_{i}-\bar{w})-\Tr (WQ)\right]^2\Neg{1.5}\big\}\label{ANA_EQN_fourth_moment}\ge0.\nonumber
	\end{align}
Then, the risk-constrained LQR problem~\eqref{FOR_EQN_LQR_constrained} is well-defined and equivalent to the sequential variational QCQP
\begin{equation}\label{ANA_EQN_Reformulation}
\begin{aligned}
    \min_{u} &\hspace{-6pt}& J(u) \triangleq\,\,&  \E\set{ x'_NQ x_N+\sum_{t=0}^{N-1} x'_tQx_t+u'_tRu_t}\\
\mathrm{s.t.} &\hspace{-6pt}& J_R(u) \triangleq\,\,& \E\set{\sum_{t=1}^{N}4x'_tQWQx_t+4x_t'QM_3}\le \bar{\epsilon}\\
&\hspace{-6pt}& & x_{t+1}=Ax_t+Bu_t+w_{t+1}\\
&\hspace{-6pt}& & u_t\in \LL_{2}(\F_t),\quad t=0,\dots N-1,
\end{aligned}\,\,,
\end{equation}
where $\bar{\epsilon}\triangleq\epsilon-Nm_4+4N\Tr\set{(WQ)^2}$.
\end{proposition}
The proof can be found in the Appendix. Proposition~\ref{ANA_PROP_Reformulation}
is critical because it shows that our risk constraint is quadratic with respect to the state and control inputs. This enables us to apply duality theory, as discussed next.
\subsection{Lagrangian Duality}
To tackle problem~\eqref{FOR_EQN_LQR_constrained}, we now consider the \textit{variational Lagrangian} $\LAG : \LL_{2}(\F_0) \times \dots \times \LL_{2}(\F_{N-1}) \times \mathbb{R}_+ \rightarrow \mathbb{R}$ of the sequential QCQP~\eqref{ANA_EQN_Reformulation}, defined as
\begin{align}\label{ANA_EQN_Lagrangian_Original}
&\LAG(u,\mu)\triangleq J(u)+\mu J_R(u)-\mu \bar{\epsilon},
\end{align}
where $\mu\in\mathbb{R}_+$ is a multiplier associated with the variational risk constraint of \eqref{ANA_EQN_Reformulation}. Hereafter, problem \eqref{ANA_EQN_Reformulation} will be called the \textit{primal problem}.
Accordingly, the \textit{dual function} $D\hspace{-0.5pt}\hspace{-0.5pt}:\hspace{-0.5pt}\hspace{-0.5pt}\mathbb{R}_{+}\hspace{-0.5pt}\hspace{-0.5pt}\rightarrow\hspace{-0.5pt}\hspace{-0.5pt}\hspace{-0.5pt}[-\infty,\infty)$
is additionally defined as
\begin{equation}\label{FDUAL}
D(\mu)\triangleq\inf_{u\in{\cal U}_0}\mathpzc{L}(u,\mu),
\end{equation}
where the \textit{implicit feasible set} $\mathcal{U}_0$ obeys ($k\le N-1$)
\begin{equation}
\mathcal{U}_k \Neg{1} \triangleq \Neg{1.5} \set{\Neg{.5} u_{k:N-1} \Neg{1}\in\Neg{1} \prod_{t=k}^{N-1} \LL_{2}(\F_t) \Neg{-1}\Bigg|\Neg{-1} x_{t+1}\Neg{1}=\Neg{1}Ax_t\Neg{1}+\Neg{1}Bu_t\Neg{1}+\Neg{1}w_{t+1} \Neg{.5}}\Neg{1}\Neg{.5},\Neg{1}\nonumber
\end{equation}
and contains the constraints of \eqref{ANA_EQN_Reformulation} that have not been dualized in the construction of the Lagrangian in \eqref{ANA_EQN_Lagrangian_Original}.
Note that it is always the case that $D\le J^{*}$ on $\mathbb{R}_{+}$, where  $J^{*}\hspace{-1pt}\in\hspace{-1pt}[0,\infty]$
denotes the optimal value of the primal problem \eqref{ANA_EQN_Reformulation}.
Then, the optimal value of the always concave \textit{dual problem}
\begin{equation}\label{DUAL}
\sup_{\mu\ge0} D(\mu) \equiv \sup_{\mu\ge0}\inf_{u\in{\cal U}_0}\mathpzc{L}(u,\mu),
\end{equation}
${D}^{*}\hspace{-2pt}\triangleq\hspace{-.5pt}\sup_{\mu\ge0}{D}(\mu)\hspace{-1pt}\in\hspace{-1pt}[-\infty,\hspace{-0.5pt}\infty]$,
is the tightest under-estimate of ${J}^{*}$, when knowing
only ${D}$.

Leveraging Lagrangian duality, we may now state the following result, which provides sufficient optimality conditions for the QCQP \eqref{ANA_EQN_Reformulation}. The proof is omitted, as it follows as direct application of \cite[Theorem 4.10]{ruszczynski2011nonlinear} 
.

\begin{theorem}[\textbf{Optimality Conditions}]\label{KKT}
Let Assumption \ref{FOR_ASS_noise} be in effect. Suppose that there exists a feasible policy-multiplier pair $(u^*,\mu^*) \in \mathcal{U}_0 \times \mathbb{R}_+$ such that
\begin{enumerate}
    \item $\LAG(u^*(\mu^*),\mu^*)=\min_{u\in\mathcal{U}_0}\LAG(u,\mu^*)=D(\mu^*)$;
    
    \item $J_R(u^*)\le\bar{\epsilon}$, i.e., the dualized variational risk constraint of \eqref{ANA_EQN_Reformulation} is satisfied by control policy $u^*$;
    
    \item $\mu^*\Neg{.5}(J_R(u^*)-\bar{\epsilon})\Neg{2.1}=\Neg{1.8}0$, i.e., complementary slackness holds.
\end{enumerate}
Then, $u^*$ is optimal for both the primal problem \eqref{ANA_EQN_Reformulation} and the initial problem \eqref{FOR_EQN_LQR_constrained}, $\mu^*$ is optimal for the dual problem \eqref{DUAL}, and \eqref{ANA_EQN_Reformulation} exhibits zero duality gap, that is, $D^*\equiv P^*<\infty$.
\end{theorem}

Theorem \ref{KKT} will be serving as the backbone of our analysis towards the solution to problem \eqref{ANA_EQN_Reformulation}, presented 
as follows.

\subsection{Optimal Risk-Aware Control Policies}
Let $\mu\ge0$ be arbitrary but fixed. First, we may simplify the form of the Lagrangian $\LAG$ and express it within a canonical dynamic programming framework. In this respect, we have the following, almost obvious, but important result.
\begin{lemma}[\textbf{Lagrangian Reformulation}]\label{ANA_THM_Lagrangian}
Let Assumption \ref{FOR_ASS_noise} be in effect and define the inflated state penalty matrix
\begin{equation}
    Q_{\mu} \triangleq Q+4\mu QWQ.\nonumber
\end{equation}
Then, for every $u_t\in\LL_2(\F_t)$, $t\le N-1$, the Lagrangian function $\LAG$ can be expressed as
	\begin{equation}\label{ANA_EQN_Lagrangian}
	\LL(u,\mu)\hspace{-1pt}=\hspace{-1pt}  \E\hspace{-1pt}\set{g_{N}(x_N,\mu)\hspace{-1pt}+\hspace{-1pt}\sum_{t=0}^{N-1}g_{t}(x_t,u_t,\mu)\hspace{-1pt}}\hspace{-.5pt}+g(\mu),\hspace{-1pt}
	\end{equation}
	where, in the notation of Proposition~\ref{ANA_PROP_Reformulation},
	\begin{align*}
		g_N(x_N,\mu)&\triangleq x_N'Q_{\mu}x_N+4\mu M_3'Qx_N,\\
	g_{t}(x_t,u_t,\mu)&\triangleq x_t'Q_{\mu}x_t+4\mu M_3'Qx_t+u_t'Ru_t, \, t\le N-1,\Neg{-1.5} \\\
	\textrm{and}\quad g(\mu)&\triangleq -\mu \bar{\epsilon}-4\mu x_0'QWQ x_0-4\mu M_3'Qx_0.
	\end{align*}
\end{lemma}
\begin{proof}
It follows from Proposition~\ref{ANA_PROP_Reformulation} and the form of $\LAG$.
\end{proof}
\begin{remark}[\textit{Relation to generalized LQR with tracking}]\label{ANA_REM_LQR}
The Lagrangian~\eqref{ANA_EQN_Lagrangian} has the structure of a generalized LQR problem with a tracking objective. By completing the squares we can rewrite the stage cost $g_t(x_t,u_t,\mu)$ as
\begin{align*}
    g_{t}(x_t,u_t,\mu)=(x_t+2\mu M_3)'Q(x_t+2\mu M_3)\Neg{-12}\\
+x'_t(4\mu QWQ)x_t+u_t'Ru_t-4\mu^2 M_3'QM_3,
\end{align*}
i.e., the state penalty is quadratic and consists of two distinct terms. The first one, i.e., $(x_t+2\mu M_3)'Q(x_t+2\mu M_3)$ is a tracking error term that forces the state to be close to the static target $-2\mu M_3$. Informally, in the case of skewed noise, by tracking $-2\mu M_3$ we pre-compensate for directions in which the distribution of the noise has heavy tails. This decreases the statistical variability of the predicted stage cost.  The second term, $x_t'(4\mu QWQ)x_t$, is a standard quadratic penalty term;  notice that, contrary to the risk-neutral case, the covariance of the noise $W$ now affects the penalty term. Informally, this term penalizes state directions which not only lead to high cost but are also more sensitive to noise, as captured by the product $QWQ$. 
Hence, the risk-neutral LQR framework can exhibit inherent risk-averse properties, provided that its parameters are selected in a principled way. Of course, selecting those parameters \emph{a priori} is not trivial.
\end{remark}

The structure of the Lagrangian as suggested by Lemma~\ref{ANA_THM_Lagrangian} enables us to derive both a closed-form expression for its minimum 
and an explicit optimal control policy. To this end, define the \textit{optimal cost-to-go} at stage $k\le N-1$ as
\begin{equation}
\begin{aligned}
& \Neg{6}\LAG^*_k(x_k,\mu) \\
& \Neg{1.5}\triangleq\Neg{1.5}\inf_{u_{k:N-1} \in \Neg{-1} \mathcal{U}_k} \E\Neg{1} \set{g_N(x_N,\mu)\Neg{1}+\Neg{1}\sum_{t=k}^{N-1}g_t(x_t,u_t,\mu)\Bigg|\F_k},\nonumber
\end{aligned}
\end{equation}
where we omit the constant components of the Lagrangian. Under this definition, it is true that
\[
D(\mu)\equiv\inf_{u\in\mathcal{U}_0} \LAG(u,\mu)=\LAG^*_0(x_0,\mu)+g(\mu).
\]
We may now derive the complete solution to \eqref{FDUAL}, which is one of the main results of this paper, and provides optimal risk-aware control policies for every fixed multiplier $\mu\ge0$.
\begin{theorem}[\textbf{Optimal Risk-Aware Controls}]\label{ANA_THM_Optimal_Input}
    Let Assumption~\ref{FOR_ASS_noise} be in effect, choose $\mu\ge0$, and adopt the notation of Lemma \ref{ANA_THM_Lagrangian}. For $t\le N-1$, the optimal cost-to-go $\LAG^*_t(x_t,\mu)$ may be expressed as
\begin{align}\label{FOR_EQN_optimal_cost_to_go}
\LAG^*_t(x_t,\mu)= x'_{t}V_{t}x_{t} +4\mu M'_3S'_{t}x_{t}+2\bar{w}'T'_t x_t+c_{t},
\end{align}
where the quantities $V_t$, $S_t$, $T_t$ and $c_t$ are evaluated through the backward recursions
	\begin{align}
\label{ANA_EQN_DEF_Riccati_Matrix}
V_{t-1}&\Neg{1}=\Neg{1}A'V_{t}A\hspace{-.5bp}+\hspace{-.5bp}Q_{\mu}\hspace{-1.5bp}-\hspace{-1.5bp}A'V_{t}B(B'V_{t}B\hspace{-.5bp}+\hspace{-.5bp}R)^{-1}B'V_{t}A,\\
K_{t-1}&\Neg{1}=\Neg{1}-(B'V_{t}B+R)^{-1}B'V_{t}A,\\
S_{t-1}	&\Neg{1}=\Neg{1}(A+BK_{t-1})'S_t+Q,\\
T_{t-1}	&\Neg{1}=\Neg{1}(A+BK_{t-1})'(T_t+V_t),\\
l_{t-1}&\Neg{1}=\Neg{1}-2\mu(B'V_{t}B+R)^{-1}B' S_t  M_3,\\
h_{t-1}&\Neg{1}=\Neg{1}-(B'V_{t}B+R)^{-1}B'(V_t+T_t)\bar{w}\quad\textrm{and}\\
c_{t-1}&\Neg{1}=\Neg{0.5}c_t+\Tr (WV_{t})+\bar{w}'(2T_{t}'+V_t)\bar{w}+4\mu M_3'S_N'\bar{w}\nonumber \\
&\Neg{-6}-(l_{t-1}+h_{t-1})'(B'V_tB+R)(l_{t-1}+h_{t-1}),
\end{align}
with terminal values $V_N=Q_{\mu}$, $S_N=Q$, $T_N=0$ and $c_N=0$.
Additionally, an optimal control policy 
that achieves the dual value in \eqref{FDUAL} may be expressed as
		\begin{align}\label{ANA_EQN_Optimal_Input}
u^*_{t}(\mu)=K_{t}x_{t}+l_{t}+h_t \in  \LL_2(\F_t),\,\,\,\forall t\le N-1,
	\end{align}
and is unique up to sets of probability measure zero. 
\end{theorem}
\begin{proof}
By using dynamic programming and assuming (temporarily) that involved measurability issues are resolved (\cite{bertsekas2012approximate}, Appendix A),  we have, for every $k\le N-1$, the recursive optimality condition (i.e., the Bellman equation)
\begin{align}
\Neg{8.5}\LAG^*_k\Neg{0.5}(x_k,\Neg{0.5}\mu)&\Neg{1.5}=\Neg{1.5}\inf_{u_k} \Neg{0.5}\E\Neg{1.8}\set{g_{k}\Neg{0.5}(x_k,\Neg{0.5}u_k,\Neg{0.5}\mu)\Neg{2}+\Neg{2}\LAG^*_{k+1}\Neg{0.8}(x_{k+1},\Neg{0.5}\mu)\big|\F_{k}\Neg{0.6}}
\Neg{1.5},\Neg{6}\nonumber
\end{align}
where minimization is taken \textit{pointwise} (i.e., over constants) over the control space $\mathbb{R}^p$, and with $\LAG^*_N(x_N,\mu)=x'_NQx_N$.
We will prove the result by induction. The base case $k=N$ is immediate. Assume that~\eqref{FOR_EQN_optimal_cost_to_go} is true for $k=t+1$;
we will show that this implies the same for $k=t$.
By Lemma~\ref{ANA_THM_Lagrangian}, and after standard algebraic manipulations, we obtain
\begin{align}
&\E\Neg{1.5}\set{g_{t}(x_t,u_t,\mu)\Neg{1.5}+\Neg{1.5}\LAG^*_{t+1}(x_{t+1},\mu)\big|\F_{t}}
\Neg{2}=\Neg{1}u_{t}'(B'V_{t+1}B\Neg{1}+\Neg{1}R) u_{t}\nonumber\\&+2\set{x'_t A' V'_{t+1}+\bar{w}'(V_{t+1}+T'_{t+1})+2\mu M'_3S'_{t+1}}Bu_{t}\nonumber\\
&+x_{t}'(A'V_{t+1}A +\bar{Q})x_{t}+c_{t+1}\nonumber\\&+2\set{2\mu M'_3S'_{t+1}+\bar{w}'(V_{t+1}+T_{t+1}')}Ax_{t}\nonumber\\
&+\bar{w}'(2T_{t}'+V_t)\bar{w}+4\mu M_3'S_N'\bar{w}+\Tr (WV_{t+1}),\label{ANA_EQN_Cost_To_Go_Form}
\end{align}
where we have exploited the identities $x_{t+1}=Ax_t+Bu_t+w_{t+1}$,  $\E(w_{t+1}|\F_{t})=\bar{w}$ and $\E (w_{t+1}-\bar{w})(w_{t+1}-\bar{w})'=W$. The reader may also verify that all measurability issues are now resolved in a recursive way, retrospectively.
The unique stationary point of the convex quadratic \eqref{ANA_EQN_Cost_To_Go_Form} is
\begin{align*}
u^*_t=K_{t}x_t+l_t+h_t,
\end{align*}
which may be readily verified to lie in $\LL_2(\F_t)$, as well.
Replacing $u^*_t$ into~\eqref{ANA_EQN_Cost_To_Go_Form} yields the optimal cost-to-go~\eqref{FOR_EQN_optimal_cost_to_go}.
\end{proof}
As suggested by Remark~\ref{ANA_REM_LQR}, it turns out that the optimal controller~\eqref{ANA_EQN_Optimal_Input} is affine with respect to the state. The noise-dependent term $\ell_t$ forces the state to track the reference $-2\mu M_3$, which points away from 
heavy-tailed regions of the noise distribution.
Meanwhile, the state-feedback term accounts for the internal dynamics. Similar to the risk-neutral case, the controller's behavior is governed the Riccati difference equation~\eqref{ANA_EQN_DEF_Riccati_Matrix}. However, we now have an inflated stage cost matrix $\bar{Q}=Q+4\mu QWQ$, instead of the original.  As suggested by the product $QWQ$, the risk-aware control gain becomes more strict in directions that are simultaneously more costly and prone to noise, as captured by the covariance $W$. Finally, the term $h_t$ acts against the mean value of the noise--such a term also appears in risk-neutral LQR.

As a corollary, from standard LQR theory, we immediately obtain that for any $\mu\ge 0$, the optimal controller~\eqref{ANA_EQN_Optimal_Input} will be internally stable, i.e., the spectral radius will be bounded $\rho(A+BK_t)<1$, as the horizon $N$ grows to infinity.
\begin{corollary}[\textbf{Internal Stability}]\label{STABILITY}
Let Assumptions \ref{FOR_ASS_noise} and \ref{FOR_ASS_LQR} be in effect, and adopt the notation of Lemma \ref{ANA_THM_Lagrangian}.
For fixed $\mu\ge0$, consider the control policy $u^*(\mu)$, as defined in~\eqref{ANA_EQN_Optimal_Input}. As $N\rightarrow \infty$, $V_t$ converges exponentially fast to the unique stabilizing solution\footnote{A stabilizing solution renders $A+BK$ stable.} of the algebraic Riccati equation
\[
V=A'VA+Q_{\mu}- A'VB(B'VB+R)^{-1}B'VA.
\] 
As a result, for every $t\ge 0$, it is true that, as $N\rightarrow \infty$, 
\begin{align*}
    K_t&\rightarrow K\triangleq -(B'VB+R)^{-1}B'VA,\\
    S_{t}&\rightarrow S \triangleq (I-(A+BK)')^{-1}Q,\\
    T_t&\rightarrow T \triangleq (I-(A+BK)')^{-1}(A+BK)'V,\\
    l_t&\rightarrow l \triangleq -2\mu(B'VB+R)^{-1}B'SM_3\quad \textrm{and}\\
    h_t&\rightarrow h \triangleq -(B'VB+R)^{-1}B'(V+T)\bar{w},
\end{align*}
exponentially fast, and the closed-loop matrix $A+BK$ is stable (spectral radius $\rho(A+BK)<1$).

\end{corollary}
\begin{proof}
Since $Q_{\mu}\succeq Q$ and $(A,Q^{1/2})$ is detectable, the pair $(A,Q_{\mu}^{1/2})$ is also detectable.
Since $(A,B)$ is stabilizable, $(A,Q_{\mu}^{1/2})$ is detectable, and $R\succ 0$, the exponential convergence of $V_t$ and $K_t$ to $V$ and $K$ respectively, and the stability of $A+BK$ follow from standard LQR theory~\cite[Chapter 4]{anderson2005optimal}.  The proof of the convergence of the remaining terms follows similar steps.
\end{proof}
Up to now we have discussed the properties of the optimal controller given a fixed $\mu\ge0$. In what follows, we show how to  compute an optimal multiplier $\mu^*$, which will also satisfy the sufficient conditions for optimality suggested by Theorem \ref{KKT}, as stated previously.

\subsection{Recovery of Primal-Optimal Solutions}
From Theorem~\ref{ANA_THM_Optimal_Input}, we know how to compute the relaxed optimal input $u^*(\mu)$, for any given multiplier $\mu\ge0$.
But the risk constraint of the primal problem \eqref{ANA_EQN_Reformulation} is the only one that has been dualized in the construction of the Lagrangian in~\eqref{ANA_EQN_Lagrangian_Original}. Then, it turns out that we can also compute an optimal multiplier $\mu^*$ via bisection, thus providing a complete solution to the primal problem.
We exploit the fact that, under the relaxed optimal policy $u^*(\cdot)$, both the LQR cost $J(u^*(\cdot))$ and the risk  functional $J_R(u^*(\cdot))$ are monotone functions.
\begin{theorem}[\textbf{Primal-Optimal Solution}]\label{OPTIMAL} Let Assumption \ref{FOR_ASS_noise} be in effect, and consider
the control policy $u^*(\mu), \mu\ge0$, as defined in~\eqref{ANA_EQN_Optimal_Input}. Then, the following statements are true:
\begin{enumerate}
    \item The LQR cost $J(u^*(\mu))$ is increasing with $\mu\ge0$, while the risk constraint functional $J_R(u^*(\mu))$ is decreasing.
    \item
    Define the multiplier  
    \begin{align}\label{ANA_EQN_Optimal_Multiplier}
    \mu^* \triangleq \inf\set{\mu \ge 0:\:J_R(u^*(\mu))\le \bar{\epsilon}}.
\end{align}
If $\mu^*$ is finite, then the policy $u^*(\mu^*)$ is optimal for the primal problem~\eqref{ANA_EQN_Reformulation}, and this is the case as long as 
\eqref{ANA_EQN_Reformulation} 
 satisfies
Slater's condition.
\end{enumerate} 
\end{theorem}
\begin{proof}
To prove part 1), let $\mu_2> \mu_1\ge 0$. From the definition of the Lagrangian and optimality of the controller $u^*(\mu)$, we obtain the inequalities
\begin{align*}
    J(u^*(\mu_1))+\mu_1 J_R(u^*(\mu_1))&\le  J(u^*(\mu_2))+\mu_1 J_R(u^*(\mu_2))\\
    J(u^*(\mu_1))+\mu_2 J_R(u^*(\mu_1))&\ge J(u^*(\mu_2))+\mu_2 J_R(u^*(\mu_2))  .
\end{align*}
By subtracting, we get
\[
(\mu_2-\mu_1)\set{J_R(u^*(\mu_1))-J_R(u^*(\mu_2))}\ge 0,
\]
which shows that $J_R(u^*(\mu_1))\ge J_R(u^*(\mu_2))$. The proof of
$J(u^*(\mu_1))\le J(u^*(\mu_2))$ is similar.

To prove part 2), we first show that, whenever $\mu^*<\infty$, $\mu^*\Neg{0.6}(J_R(u^*\Neg{0.6}(\mu^*))\Neg{1}-\Neg{1}\bar{\epsilon})\Neg{2.2}=\Neg{2}0$, i.e., complementary slackness holds. 
We have two cases: either $\mu^*=0$, where complementary slackness is satisfied trivially; or $\mu^*>0$, $J_R(u^*(\mu^*))\le \bar{\epsilon}$. Therefore, it will be sufficient to show that in the latter case we can only have $J_R(u^*(\mu^*))= \bar{\epsilon}$.
Since $\mu^*>0$, it is true that $J_R(u^*(0))>\bar{\epsilon}$. From Theorem~\ref{ANA_THM_Optimal_Input} and Proposition~\ref{ANA_PROP_Risk_Evaluation}, it also follows that the function $J_R(u^*(\mu))$ is continuous with respect to $\mu$ (all matrix inverses in~\eqref{ANA_EQN_DEF_Riccati_Matrix} are continuous since $R\succ 0$). Now, assume that $J_R(u^*(\mu^*))< \bar{\epsilon}$. Then by continuity, there exists a $0<\bar{\mu}<\mu^*$ such that $J_R(u^*(\bar{\mu}))=\bar{\epsilon}$, contradicting the definition of $\mu^*$. Hence, we can only have $J_R(u^*(\mu^*))= \bar{\epsilon}$, which shows that complementary slackness is satisfied. 

Now, complementary slackness, along with the trivial fact that $J_R(u^*(\mu^*))\le \bar{\epsilon}$ imply that the policy-multiplier pair $(u^*(\mu^*),\mu^*) \in \mathcal{U}_0 \times \mathbb{R}_+$ satisfies the sufficient conditions for optimality provided by Theorem \ref{KKT}. Enough said. 

To prove the last claim of part 2), suppose that \eqref{ANA_EQN_Reformulation} 
 satisfies
Slater's condition, i.e., that there is an admissible policy $u^\dagger$ in $\mathcal{U}_0$ such that $J_R(u^\dagger)-\bar{\epsilon}<0$. For every $\mu\ge0$, we have
\begin{equation}
\begin{aligned}
    D(\mu) &\le J(u^\dagger) + \mu (J_R(u^\dagger)-\bar{\epsilon}) \nonumber\\
\implies D(\mu) - \mu (J_R(u^\dagger)-\bar{\epsilon}) &\le J(u^\dagger)<\infty. \nonumber
\end{aligned}
\end{equation}
Next, suppose that, for every $\mu\ge0$, $J_R(u^*(\mu))-\bar{\epsilon}\ge0$. Because $J(u^*(\cdot))$ is increasing on $\mathbb{R}_+$, it must be true that
\begin{equation}
\begin{aligned}
    J(u^\dagger) \Neg{2}&\ge\Neg{1} \sup_{\mu\ge0} D(\mu) \Neg{1}-\Neg{1} \mu (J_R(u^\dagger)\Neg{1}-\Neg{1}\bar{\epsilon}) \nonumber\\
    & =\Neg{1} \sup_{\mu\ge0} J(u^*(\mu)) \Neg{1}+\Neg{1} \mu (J_R(u^*(\mu))\Neg{1}-\Neg{1}\bar{\epsilon}) \Neg{1}-\Neg{1} \mu (J_R(u^\dagger)\Neg{1}-\Neg{1}\bar{\epsilon}) \nonumber\\
    & = \Neg{1}\infty,
\end{aligned}
\end{equation}
which contradicts the fact that $J(u^\dagger)<\infty$. Therefore, there must exist $\mu^\dagger \ge0$, such that $J_R(u^*(\mu^\dagger))-\bar{\epsilon}<0$. But $J_R(u^*(\cdot))$ is decreasing on $\mathbb{R}_+$ and, consequently, it must be the case that $\mu^*\in[0,\mu^\dagger)$. The proof is now complete.
\end{proof}

Theorem \ref{OPTIMAL} implies that we can find an optimal  multiplier satisfying the optimality conditions on Theorem \ref{KKT} from~\eqref{ANA_EQN_Optimal_Multiplier}, by performing simple bisection on $\mu$.
Of course, this requires evaluating the risk constraint functional $J_R(u^*(\mu))$ for different values $\mu\ge0$. The evaluation may be performed in a recursive fashion, as the following result suggests.
\begin{proposition}[\textbf{Risk Functional Evaluation}]\label{ANA_PROP_Risk_Evaluation} Let Assumption \ref{FOR_ASS_noise} be in effect, and adopt the notation of Lemma \ref{ANA_THM_Lagrangian}.
For fixed $\mu\ge0$, consider the control policy $u^*(\mu)$, as defined in~\eqref{ANA_EQN_Optimal_Input}.
With terminal values $P_N=4QWQ$, $z_N=4M_3'Q$, $r_N=0$, consider the backward recursions
\begin{align}
&P_{t-1}=(A+BK_{t-1})'P_t(A+BK_{t-1})+4QWQ,\nonumber\\
&z_{t-1}=(A+BK_{t-1})'z_t+4QM_3\nonumber\\
&\,\,\,+2(A+BK_{t-1})'P_t\paren{Bh_{t-1}+Bl_{t-1}+\bar{w}}\,\,\,\, \textrm{and}\nonumber\\
&r_{t-1}=r_t+\Tr(P_tW)+z_t'(\bar{w}+Bl_{t-1}+Bh_{t-1})\nonumber\\
&\,\,\,+(Bl_{t-1}+Bh_{t-1}+\bar{w})'P_{t}(Bl_{t-1}+Bh_{t-1}+\bar{w}).\nonumber
\end{align}
Then, the risk constraint in problem \eqref{ANA_EQN_Reformulation} may be evaluated by
\begin{equation}
J_R(u^*(\mu))=x'_0P_0x_0+z'_0x_0+r_0.\nonumber
\end{equation}
\end{proposition}
\begin{proof}
Omitted; it is similar to that of Theorem~\ref{ANA_THM_Optimal_Input}. 
\end{proof}

%% file: Simulations.tex
\section{Simulations And Discussion}\label{Section_Simulations}
\begin{figure}[t]
	\vspace{3.5bp}
	\centering
	\includegraphics[width=\columnwidth]{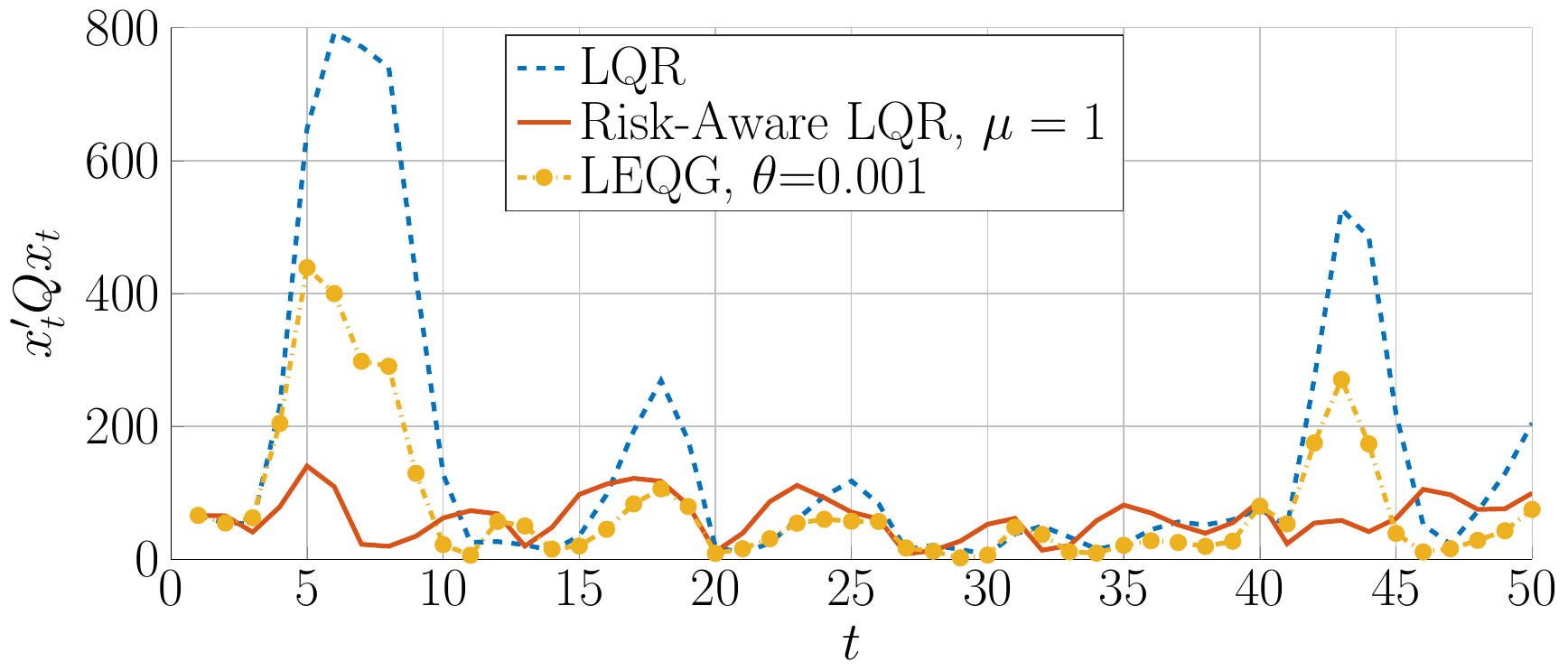}\vspace{-4bp}
	\caption{Evolution of the state penalties $x'_kQx_k$, over the first $50$ steps. Notice that our risk-aware LQR controller indeed limits the variability of $x'_kQx_k$. In fact, it sacrifices performance under small wind forces, but protects the system against large wind gusts, for example at time $5-10$.}
	\vspace{-12bp}
	\label{fig:xQx}
\end{figure}
Consider a flying robot that moves on a horizontal plane, i.e., the Euclidean space $\mathbb{R}^2$. We assume that its linearized dynamics can be abstracted by a double integrator as
\begin{equation}\label{SIM_eq:integrator}
    x_{k+1}=\matr{{cccc}1&T_s&0&0\\0&1&0&0\\0&0&1&T_s\\0&0&0&1}x_k+\matr{{cc}\tfrac{T_s^2}{2}&0\\T_s&0\\0&\tfrac{T_s^2}{2}\\0&T_s}(v_k+d_k),
\end{equation}
where  $T_s=0.5$ is the sampling time, $x_{k,1}$, $x_{k,3}$ are the position coordinates, $x_{k,2}$, $x_{k,4}$ the respective velocities and $v_k$ is the acceleration input. Let $d_k$ be a wind disturbance force that acts on the robot, which is modeled as follows: We assume that $d_{k,1}$ constitutes the dominant wind direction with non-zero mean and large variability, while the orthogonal direction $d_{k,2}$ is a weak wind direction with zero mean and small variability.
We model $d_{k,1}$ as a mixture of two gaussians $\mathcal{N}(30,30)$, $\mathcal{N}(80,60)$ with weights $0.8$ and $0.2$, respectively. This bimodal distribution models the presence of infrequent but large wind gusts.  The weak direction $d_{k,2}$ is modeled as zero-mean Gaussian~$\mathcal{N}(0,5)$.
If we cancel the mean of $d_k$ by applying $v_k=u_k-\E d_k$, then~\eqref{SIM_eq:integrator}, can be written in terms of~\eqref{FOR_EQN_system}, where $w_k=B(d_k-\E d_k)$ is now a zero-mean disturbance.

Consider now the LQR problem with parameters 
\[
Q = \textrm{diag}(1,0.1,2,0.1)
\Neg{2}\quad\textrm{and}\quad R=I,
\] 
and a horizon of length $N=5000$. We primarily compare our risk-aware LQR formulation with the classical, risk-neutral LQR via simulations. To tune our controller, we vary $\mu$ in~\eqref{ANA_EQN_Optimal_Input} directly instead of varying $\epsilon$. We also (heuristically) compare our controller with the exponential (LEQG) method, even though the noise is not Gaussian, by plugging in the second order statistics $W$. Let the tuning parameter of LEQG be $\theta$. Note that the exponential problem is well defined only if $\theta<0.001276$ (roughly), where the ``neurotic breakdown" occurs~\cite{Whittle1981}. For the purpose of comparison, we simulate all schemes under the same noise sequence $w_{0:N}$.

In Fig.~\ref{fig:xQx}, we see the evolution of the state penalty terms $x_k'Qx_k$, for the first $50$ time steps, under the different control schemes.  By slightly sacrificing performance under small wind forces, our risk-aware LQR controller forces the state to have less variability and, thus, protects the robot against large gusts. On the other hand, the state of the robot state can grow very large under the risk-neutral \textit{and} LEQG schemes. This behavior is illustrated  even more clearly in Fig.~\ref{fig:cdf}, where we present the time-empirical cumulative distribution of the state penalties for all $N$ time steps. The time-empirical "probability" of suffering large state penalties is drastically smaller compared to LQR or LEQG.

\begin{figure}[t]
\vspace{3.5bp}
	\centering{}
	\includegraphics[width=\columnwidth]{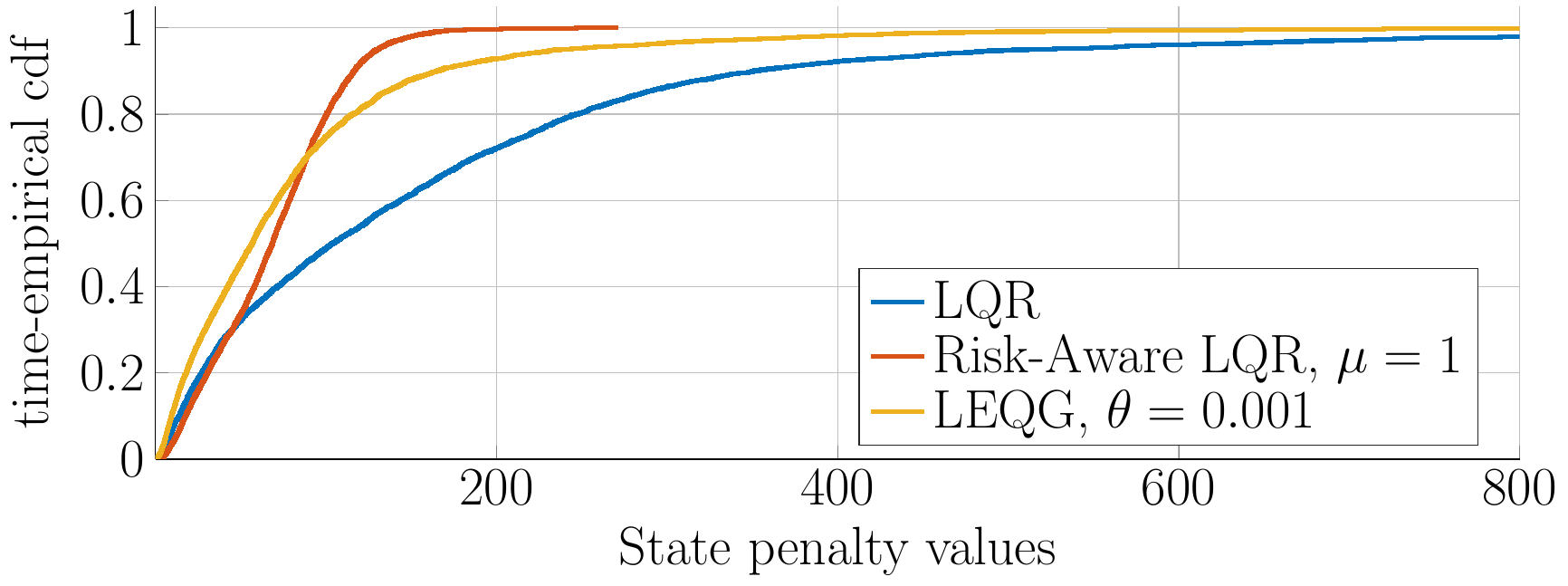}\vspace{-4bp}
	\caption{The time-empirical cdf for the state penalties $x_k'Qx_k$, $k\le N$, for the LQR (risk-neutral), our method, and LEQG. Our method sacrifices some average performance but exhibits much smaller variability for the state penalties. It also protects the system against rare but large wind gusts.}
	\vspace{-12bp}
	\label{fig:cdf}
\end{figure}
\begin{figure}[t]
\vspace{3.5bp}
	\centering
	\includegraphics[width=\columnwidth]{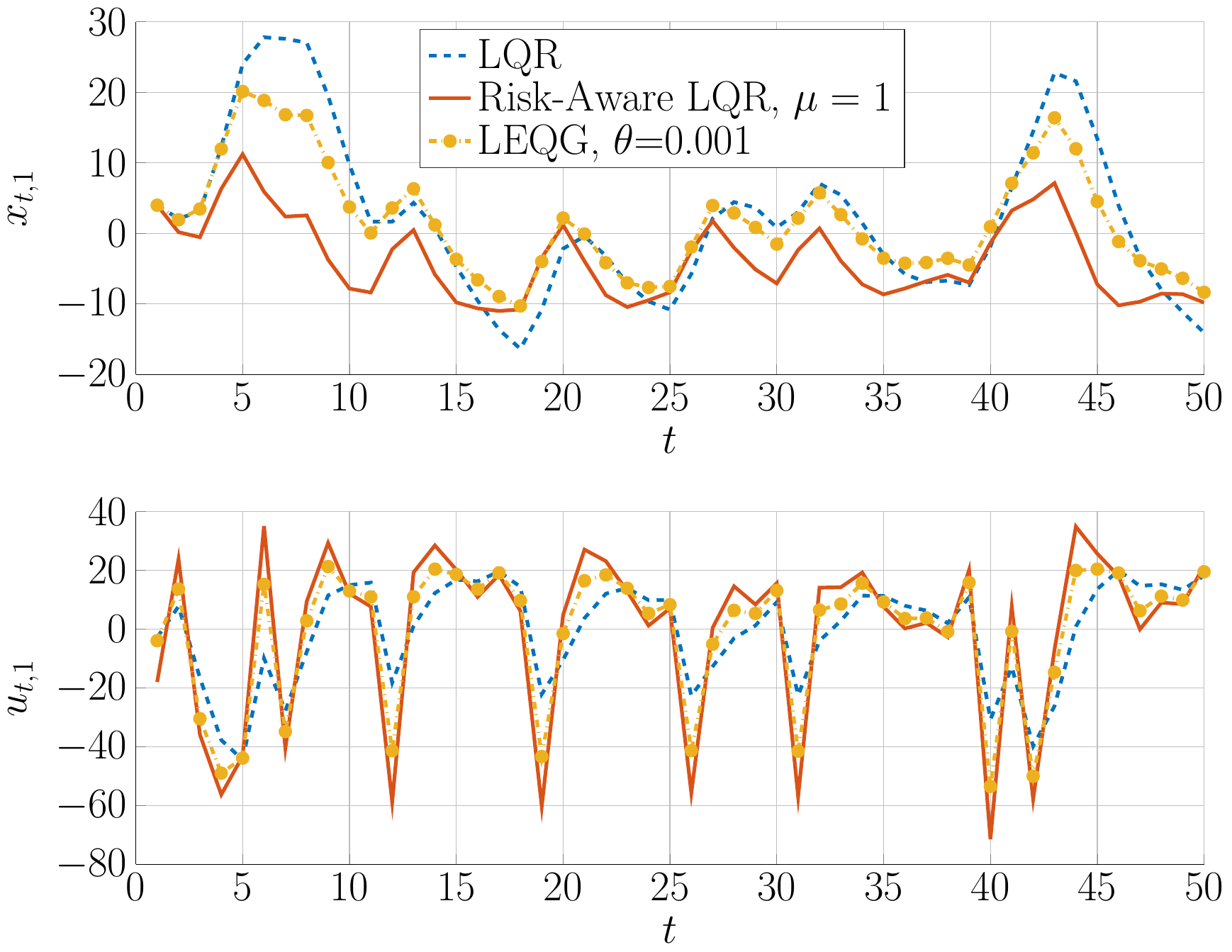}\vspace{-4bp}
	\caption{Evolution of the state $x_{k,1}$, and the input $u_{k,1}$ over the first $50$ steps. The controller pushes the state away from the direction of the large gusts, which helps the robot to avoid extreme perturbations. Meanwhile, by inflating the state penalty with the $\mu QWQ$, we force the state-feedback component to be more cautious with the state. Naturally, being more cautious with the state requires extra control effort.}
	\vspace{-13bp}
	\label{fig:x1}
\end{figure}
To better illustrate how the proposed risk-aware controller works, we also discuss the evolution of the position $x_{k,1}$ and the input $u_{k,1}$, as shown in Fig. \ref{fig:x1}, for the first $50$ steps. First, we observe that the controller pushes the state $x_{k,1}$ towards negative values, away from the direction of the large gusts.
Second, notice that we penalize $x_{k,3}$ more in $Q$. In fact, the risk-neutral LQR results in the steady state gains $K_{\mathrm{LQR},11}=-0.697$, $K_{\mathrm{LQR},12}=-1.201$, $K_{\mathrm{LQR},23}=-0.925$, $K_{\mathrm{LQR},24}=-1.376$, i.e., it is stricter with direction $x_{k,3}$. However, $x_{k,1}$ exhibits more variability due to the strong wind direction. In contrast, our risk-aware scheme adapts to the noise in a principled way. Due to the inflation term $\mu QWQ$, our scheme returns the steady-state gains $K_{11}=-2.1008$, $K_{12}=-2.2132$, $K_{23}=-1.1161$, $K_{24}=-1.5131$, which means that the risky direction $x_{k,1}$ is controlled more strictly. Naturally, being more cautious with the state leads to higher control effort. Lastly, although the LEQG controller is also more state-cautious, it is agnostic to the heavy tails of the wind distribution. Hence, it still suffers from large perturbation due to the wind gusts.

%% file: Extensions.tex
\section{Conclusion and Future Work}\label{Section_Extensions}
We presented a new risk-aware reformulation of the classical LQR problem, where we introduce a new risk measure to be used as an explicit and tunable risk constraint, along with the standard LQR objective. By restricting the expected cumulative predictive variance of the state penalties, we can decrease the variability of the state at will, protecting the system against uncommon but strong random disturbances. The optimal controller enjoys a simple closed-form expression with clear interpretation, is always stable and is easy to tune.
Our scheme works for arbitrary noise process distributions, as long as the corresponding  fourth-order moments are finite.

Moving forward, our framework opens up many directions for extensions and future research. First, we would like to note that our analysis does not depend on the constraints
having the same matrix $Q$ as in the cost. In fact, we can define our risk constraint as
\[
\E \set{\sum_{t=1}^{N} \left[ x'_tQ_c x_t-\E\paren{x'_tQ_cx_t|\F_{t-1}} \right]^2}\le \epsilon,
\]
where $Q_c$ is a design choice. Proposition~\ref{ANA_PROP_Reformulation} still holds, in the sense that the constraint can be rewritten as a quadratic one.
This implies that, thanks to its simplicity, the predictive variance constraint can be easily incorporated in more general problems, e.g., MPC or classical constrained-LQR, adding a risk-aware flavor to them.
Another possibility is to employ stage-wise constraints of the form
\[
\E\left[ x'_tQ_t x_t-\E\paren{x'_tQ_tx_t|\F_{t-1}} \right]^2\le \epsilon_t,\quad t\le N.
\]
In this case, the optimal controller will be similar to~\eqref{ANA_EQN_Optimal_Input}, but will depend on multiple Lagrange multipliers that can be optimized with primal-dual algorithms.
Lastly, our predictive variance constraint is based on one-step-ahead prediction. In some cases, careless selection of $Q_c$ 
might make our controller more myopic. At the same time, though, increasing the prediction horizon  might not always preserve the quadratic form of the constraint. In future work, we would also like to address this issue.

%% file: appendix.tex
\section*{Appendix: Proof of Proposition~\ref{ANA_PROP_Reformulation}}
Let $\Delta_t\triangleq x'_tQ x_t-\E\paren{x'_tQx_t|\F_{t-1}}$ be the prediction error of the stage penalty at time $t$ given $\F_{t-1}$.
We proceed in two steps.
	First, we show that  $\Delta_t$ is well-defined and belongs to $\LL_2(\F_{t})$. Second, we obtain the closed form expression for the expected predictive variance $\E \set{\Delta^2_t}$. 
	
	\textit{Step a).} The state $x_t$ of the system depends linearly on past inputs $u_k$ as well as past noises $w_{k+1}$, for $k\le t-1$. Under the constraint $u_k\in\LL_2(\F_k)$, and since by Assumption~\ref{FOR_ASS_noise} $w_k\in\LL_2(\F_k)$, it also follows that 
	$x_t\in \LL_2(\F_t)$, for all $t\le N-1$. As a result, the expectation of $x_t'Qx_t$ exists and any conditional expectation $\E\paren{x'_tQ x_t|\F_{t-1}}$ is well-defined and finite almost everywhere, for all $t\le N-1$.
	Define
	\begin{align}
	\hat{x}_{t}&\triangleq \E(x_t|\F_{t-1})=Ax_{t-1}+Bu_{t-1}+\bar{w}\quad\textrm{and}\\
	\delta_t&\triangleq w_t-\bar{w}.
	\end{align}
	Note that $\hat{x}_t$ is well-defined since $x_{t}\in \LL_2(\F_t)$.
	Replacing $x_t$ with $\hat{x}_{t}+\delta_{t}$, we obtain the representation
	\begin{align*}
	x'_tQ x_t&=\hat{x}'_{t}Q\hat{x}_{t}+2\hat{x}'_{t}Q\delta_t+\delta_t'Q\delta_t.
	\end{align*}
	All of the terms above are integrable since $\hat{x}_t$, $\delta_t$ are square-integrable.
	Since $\hat{x}_{t}$ is measurable with respect to $\F_{t-1}$, the expectation of $x'_tQ x_t$ conditioned on $\F_{t-1}$ is
		\begin{align*}
	\E\paren{x'_tQ x_t|\F_{t-1}}&=\hat{x}'_{t}Q\hat{x}_{t}+\Tr (WQ).
	\end{align*}
Then, the difference of the above quantities is
	\begin{align*}
\Delta_t=\delta_t'Q\delta_t -\Tr (WQ)+2\hat{x}'_{t}Q\delta_{t}.
	\end{align*}
	Computing the squares of both sides leads to the expression
	\begin{equation}\label{ANA_EQN_Prediction_Difference}
	\begin{aligned}
\Delta^2_t&=(\delta_t'Q\delta_t -\Tr (WQ))^2+4\hat{x}'_{t}Q\delta_{t}\delta'_{t}Q\hat{x}_t\\&\quad\,+4\hat{x}'_{t}Q\delta_{t}(\delta_t'Q\delta_t -\Tr (WQ)).
	\end{aligned}
	\end{equation}
	Note that all of the above terms are integrable, hence $\E\set{\Delta^2_t}$ is well-defined and finite. Integrability of the first term follows from Assumption~\ref{FOR_ASS_noise}. Integrability of the second term comes from the fact that $\hat{x}_t$ and $\delta_t$ are square-integrable and independent of each other. Similarly, itegrability of the last term follows from integrability of $\hat{x}_t$, Assumption~\ref{FOR_ASS_noise} 
	and independence of $\hat{x}_t,\delta_t$. 
	
	\textit{Step b).}
From~\eqref{ANA_EQN_Prediction_Difference}, it is true that
	\begin{align*}
	&\E\set{\Delta_t^2|\F_{t-1}}=
	4\hat{x}'_{t}QWQ\hat{x}_{t}+m_4+4\hat{x}'_{t}QM_3.
	\end{align*}
  Taking expectation again gives
  \begin{align*}
	\E\set{\Delta_t^2}=m_4+\E(4\hat{x}'_{t}QWQ\hat{x}_{t}+4\hat{x}'_{t}QM_3).
  \end{align*}
By orthogonality of $\hat{x}_t$, $\delta_t$, and since $\E \delta_t=0$, $\E \delta_t\delta'_t=W$, 
  \begin{align*}
	\E\hspace{-1pt}\set{\Delta_t^2}&\hspace{-2pt}=\hspace{-1pt}\E(4x_{t}'QWQx_{t}\hspace{-1pt}+\hspace{-1pt}4x_{t}'QM_3)\hspace{-1pt}+\hspace{-1pt}m_4\hspace{-1pt}-\hspace{-1pt}4\Tr\Neg{1}\set{(WQ)^2}\hspace{-1pt}.
  \end{align*}
  The result follows if we replace $\E\set{\Delta^2_t}$ with the  right-hand side above in the risk constraint
 $
  \sum_{t=1}^{N}\E\set{\Delta^2_t}\le \epsilon.
$
\hfill $\qedsymbol$